\documentclass[11pt]{article}
\usepackage{amsmath}
\usepackage{amsthm}
\usepackage{amssymb}
\usepackage{url}
\usepackage[final]{showkeys}
\usepackage[dvipsnames]{xcolor}
\usepackage{graphicx}
\usepackage{bm}

\usepackage[margin=1in]{geometry}

\usepackage{titlesec}

\titlespacing*{\section}{0pt}{11pt plus 2pt minus 2pt}{8pt plus 1pt minus 1pt}
\titlespacing*{\subsection}{0pt}{9pt plus 2pt minus 2pt}{6pt plus 1pt minus 1pt}

\usepackage[vlined]{algorithm2e}

\usepackage{fullpage}

\usepackage{stmaryrd}

\usepackage[colorlinks=true, allcolors=blue]{hyperref}

\usepackage{latexsym}

\newtheorem{theorem}{Theorem}[section]
\newtheorem{lemma}[theorem]{Lemma}

\newtheorem{corollary}[theorem]{Corollary}
\newtheorem{definition}[theorem]{Definition}

\newcommand{\AppendixBanner}{%
  \hbox to \textwidth{%
    \leaders\hbox{\rule[.55ex]{1em}{0.5pt}}\hfill
    \enspace\textbf{APPENDIX}\enspace
    \leaders\hbox{\rule[.55ex]{1em}{0.5pt}}\hfill
  }%
}

\makeatletter
\def\ps@appendix{%
  \let\@mkboth\markboth
  \def\@oddhead{\AppendixBanner}%
  \def\@evenhead{\AppendixBanner}%
  \def\@oddfoot{\hfil\thepage\hfil}%
  \def\@evenfoot{\hfil\thepage\hfil}%
}
\makeatother

\usepackage{caption}
\usepackage{subcaption}

\usepackage{ifpdf}
\ifpdf
\else
\fi

\title{An improved randomized algorithm for finding sinks of\\ Acyclic Unique Sink Orientations of the binary hypercube}

\newcommand{\ignore}[1]{}
\newcommand{\remove}[1]{}

\newcommand{\EQ}{\;=\;}
\newcommand{\LE}{\;\le\;}
\newcommand{\LT}{\;<\;}
\newcommand{\GE}{\;\ge\;}

\newcommand{\be}{\bm{e}}
\newcommand{\bs}{\bm{s}}
\newcommand{\bx}{\bm{x}}
\newcommand{\by}{\bm{y}}
\newcommand{\bz}{\bm{z}}

\newcommand{\vH}{H}

\newcommand{\cube}{\{0,1\}^n}
\newcommand{\sink}{sink}
\newcommand{\cF}{\mathcal{F}}

\newcommand{\floor}[1]{{\lfloor#1\rfloor}}
\newcommand{\ceil}[1]{{\lceil#1\rceil}}

\usepackage[
    style=alphabetic, 
    backend=biber, 
    giveninits=false, 
    maxbibnames=99, 
    maxcitenames=3, 
    maxalphanames=4,  
    minalphanames=3,
    doi=true, 
    url=true
]{biblatex}
\DeclareSourcemap{
  \maps[datatype=bibtex]{
    \map{
      \step[fieldsource=doi, final]
      \step[fieldset=url, null]
    }
  }
}

\renewbibmacro{in:}{%
  \ifentrytype{article}{}{\printtext{\bibstring{in}\intitlepunct}}%
}

\DeclareFieldFormat[article,incollection,inproceedings]{title}{\textit{#1}}

\begin{document}

\author{Thomas Dueholm Hansen \thanks{E-mail: {\tt
    thomasdueholm@gmail.com}.} \and Uri Zwick
\thanks{Blavatnik School of
Computer Science and AI, Tel Aviv University,
  Israel. E-mail: {\tt
    zwick@tau.ac.il}.}}

\date{}

\maketitle

\begin{abstract}
\setlength{\baselineskip}{12pt} %
\setlength{\parskip}{12pt} %
\noindent %
Acyclic Unique Sink Orientations (AUSOs) of the binary hypercube form an elegant abstraction of various optimization problems, ranging from linear programs whose feasibility region is a combinatorial hypercube, to 2-player turn-based stochastic games (2TBSGs) in which one of the players has two actions from each one of its states. The currently fastest known randomized algorithm for finding the sink of such AUSOs is the Random-Facet algorithm of Kalai (STOC 1992), and of Matou{\v{s}}ek, Sharir and Welzl (SoCG 1992), adapted to the binary hypercube by G{\"a}rtner (RSA 2002). The expected number of vertex queries made by the algorithm is at most $e^{2\sqrt{n}}$, where~$n$ is the dimension of the cube. We design an improved randomized sink-finding algorithm that performs an expected number of only $e^{\sqrt{(2+o(1))n}}$ queries. This is the first improved bound for the problem in over two decades. The improved algorithm is obtained by adapting improved versions of the Random-Facet algorithm for solving general linear programs, and LP-type problems, developed by Kalai (STOC 1992), G{\"a}rtner (FOCS 1992), and by Hansen and Zwick (STOC 2015). The improved bound is not implied by these much more general algorithms. A key feature of the improved algorithm and its analysis is its simplicity. The structure of binary AUSOs enables a simpler description and a much simpler analysis of the adapted algorithm. In particular, the recurrence relations bounding the complexity of the algorithm depend on only one variable, and not two, and they can be solved using fairly straightforward techniques.
\end{abstract}

\section{Introduction}\label{S-intro}

An orientation of the binary hypercube is said to be a \emph{Unique Sink Orientation (USO)} if and only if every subcube has a unique sink under this orientation. A \emph{sink} is a vertex with only incoming edges. An orientation is said to be an \emph{Acyclic Unique Sink Orientation (AUSO)} if it is also acyclic, i.e., it has no directed cycles. The systematic study of USOs and AUSOs was initiated by Szab{\'o} and Welzl \cite{SzWe01}.

As USOs and AUSOs are of exponential size in their dimension, we assume that we are given \emph{oracle} access to them. Given an edge of the hypercube we can determine its orientation, or more generally, given a vertex we can determine the direction of all edges incident on this vertex. The task is to find the sink of the USO or AUSO using a minimal number of \emph{queries} to the oracle.

Various interesting problems can be reduced to the sink-finding problem in USOs and AUSOs. G{\"a}rtner and Schurr \cite{GaSc06} showed that a polynomial-time sink-finding algorithm for USOs implies a strongly polynomial time algorithm for linear programming, a major open problem. A polynomial-time sink-finding algorithm for AUSOs implies a strongly polynomial time algorithm for solving 2-Player, Turn-Based, Stochastic Games (2-TBSGs), an intriguing problem for which even polynomial time algorithms are not currently known. (See, e.g., Ludwig \cite{Ludwig95}, Bj{\"o}rklund and Vorobyov \cite{BjVo05,BjVo07},  Halman \cite{Halman07}, Hansen, Miltersen and Zwick \cite{HaMiZw13},  and Fijalkow et al.~\cite{gamesbook}.)

Unfortunately, the fastest known algorithms for finding sinks of USOs are exponential, and the fastest known algorithm for finding sinks of AUSOs are sub-exponential. More precisely, Szab{\'o} and Welzl \cite{SzWe01} obtained a deterministic $O(1.61^n)$ algorithm and a randomized $O(1.47^n)$ algorithm for finding the sink of a USO of dimension~$n$. An improved randomized $O(1.438^n)$ algorithm was obtained by Rote (see Tessaro \cite{Tessaro04}.) G{\"a}rtner~\cite{Gartner02} obtained a randomized $e^{2\sqrt{n}}$ algorithm for finding the sink of an AUSO of dimension~$n$. G{\"a}rtner's algorithm is an adaptation of randomized linear programming algorithms of Kalai \cite{Kalai92,Kalai97}, and of Matou{\v{s}}ek, Sharir and Welzl~\cite{MaShWe96}.

Our main result is an improved randomized sink-finding algorithm for AUSOs with a running time of $e^{\sqrt{(2+o(1))n}}$. (Note that the coefficient of $\sqrt n$ in the exponent is reduced from~$2$ to essentially $\sqrt2$.) This is the first improved algorithm for the problem in over two decades.

Our algorithm is obtained by adapting improved versions of the algorithms of Kalai \cite{Kalai92,Kalai97}, and of Matou{\v{s}}ek, Sharir and Welzl \cite{MaShWe96}, for solving linear programming, or more generally LP-type problems. (For a definition of LP-type problems, see~\cite{MaShWe96}.) Such improved algorithms were obtained by Kalai \cite{Kalai92,Kalai97}, G{\"a}rtner \cite{Gartner95,Gartner95PhD} and Hansen and Zwick \cite{HaZw15}. A common feature of these improved algorithms is that they are fairly complicated, and have fairly complicated analyses, involving complicated bivariate recurrence relations. We show that by adapting the algorithm of Hansen and Zwick \cite{HaZw15} to the binary hypercube, we get a faster and much simpler algorithm, and with a much simpler analysis that uses only univariate recurrence relations. (Our algorithm only works on AUSOs and does not improve on the more general algorithm of \cite{HaZw15}.)

\subsection*{Paper organization}

The rest of the paper is organized as follows. In Section~\ref{S-AUSO} we give a formal definition of AUSOs. 
In Section~\ref{S-RF} we give the basic version of the Random-Facet algorithm, due to G{\"a}rtner \cite{Gartner02}, whose complexity is $e^{2\sqrt{n}}$. In Section~\ref{S-fiction} we describe a `fictional', non-realizable, algorithm with an improved complexity of $e^{\sqrt{2n}}$. In Section~\ref{S-improved} we use insights gained from the fictional algorithm to obtain a realizable algorithm. We also obtain a recurrence relation that bounds the complexity of the algorithm. In Section~\ref{S-solve} we solve this non-standard recurrence relation and show that the complexity of the obtained algorithm is $e^{\sqrt{(2+o(1))n}}$, essentially matching the complexity of the fictional algorithm. In Section~\ref{S-tighter} we describe a variant of the improved algorithm that avoids the need to tune a constant appearing in the algorithm of Section~\ref{S-improved}, and also provide a self-contained analysis of this algorithm. Section~\ref{S-conclude} contains some concluding remarks and open problems.

\subsection*{Related work}

Zwick \cite{Zwick26} recently obtained a version of the basic Random-Facet algorithm of G{\"a}rtner \cite{Gartner02} for \emph{non-binary} AUSOs called the \emph{Random-Action-Removal} algorithm. Using a complicated analysis, he shows that the running time of the algorithm on an AUSO of the cube $\prod_{i=1}^n A_i$ is $e^{O(\sqrt{n\log\frac{m}{n})}}$, where $m=\sum_{i=1}^n |A_i|$. The main challenge in~\cite{Zwick26} is the analysis of the non-binary version of the algorithm. No improvements are implied by this result for the binary AUSOs studied here. The current paper is, in a sense, orthogonal to \cite{Zwick26}. It deals only with binary AUSOs and it improves on the algorithm of G{\"a}rtner \cite{Gartner02} for this case using ideas taken from Kalai \cite{Kalai92,Kalai97}, G{\"a}rtner \cite{Gartner95} and Hansen and Zwick \cite{HaZw15}.


\begin{figure}
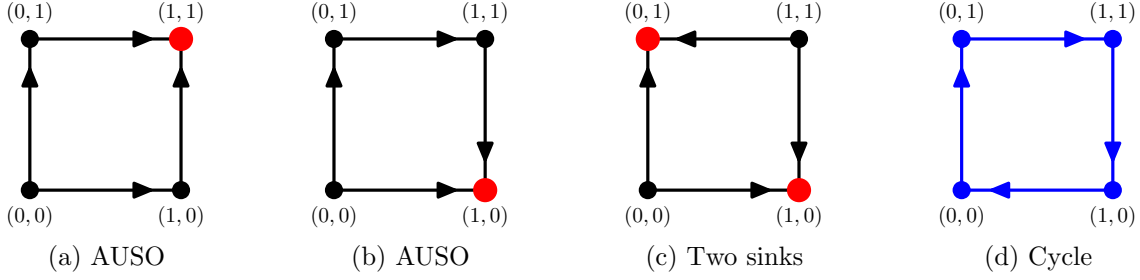

\centering
\begin{subfigure}{.24\textwidth}
  \centering
  \includegraphics[scale=0.8]{figure/auso.4}
  \caption{AUSO}
\end{subfigure}%
\begin{subfigure}{.24\textwidth}
  \centering
  \includegraphics[scale=0.8]{figure/auso.5}
  \caption{AUSO}
\end{subfigure}
\begin{subfigure}{.24\textwidth}
  \centering
  \includegraphics[scale=0.8]{figure/auso.6}
  \caption{Two sinks}
\end{subfigure}
\begin{subfigure}{.24\textwidth}
  \centering
  \includegraphics[scale=0.8]{figure/auso.7}
  \caption{Cycle}
\end{subfigure}
\captionsetup{skip=11pt}
\caption{All non-isomorphic orientations of binary 2-cubes.}\label{F-2}
\smallskip
\end{figure}


\begin{figure}
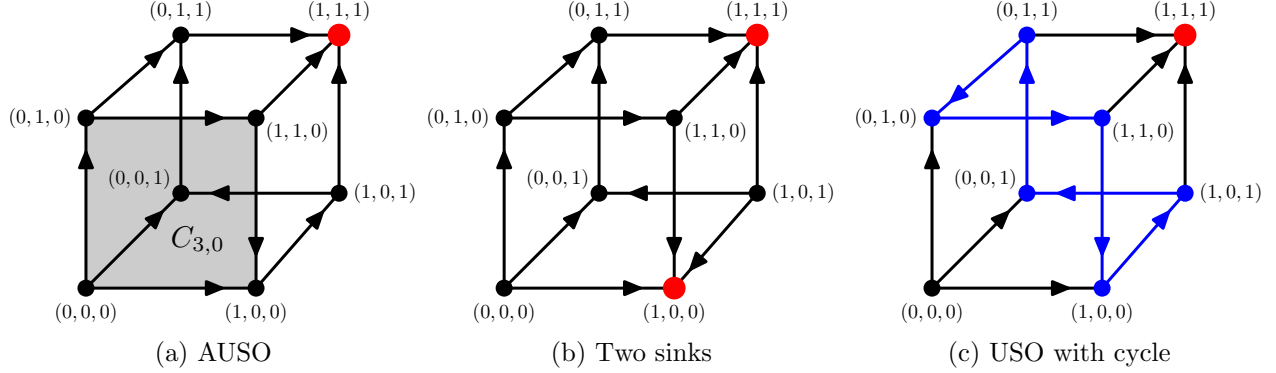

\centering
\begin{subfigure}{.33\textwidth}
  \centering
  \includegraphics[scale=0.75]{figure/auso.1}
  \caption{AUSO}
\end{subfigure}%
\begin{subfigure}{.33\textwidth}
  \centering
  \includegraphics[scale=0.75]{figure/auso.2}
  \caption{Two sinks}
\end{subfigure}
\begin{subfigure}{.33\textwidth}
  \centering
  \includegraphics[scale=0.75]{figure/auso.3}
  \caption{USO with cycle}
\end{subfigure}
\captionsetup{skip=11pt}
\caption{Some orientations of the binary 3-cube.}\label{F-3}
\smallskip
\end{figure}

\section{Acyclic Unique Sink Orientations (AUSOs)}\label{S-AUSO}

Our definitions of USOs and AUSOs follow the definitions of Szab{\'o} and Welzl \cite{SzWe01}. Similar definitions, using a different terminology, were given earlier by Williamson Hoke \cite{WiHo88}. The various definitions are illustrated in Figures~\ref{F-2} and~\ref{F-3}.

\begin{definition}[Hypercubes]
    The \emph{binary hypercube} $C=C_n=(V_n,E_n)$ of dimension~$n$ is an undirected graph whose vertex set is $V_n=\{0,1\}^n$ and whose edge set is $E_n=\{\{\bx,\by\}\mid \bx,\by\in\cube , d_H(\bx,\by)=1\}$, where $d_H(\bx,\by)$ is the Hamming distance between $\bx$ and $\by$. In other words, there is an edge between two vertices~$\bx,\by\in\cube$ if and only if they differ in exactly one coordinate. 
\end{definition}

We let $\be_i=(0,\dots,0,1,0,\dots,0)\in\cube$ be the $i$-th unit vector. Each edge of $C_n$ is thus of the form $\{\bx,\bx\oplus \be_i\}$, for some $\bx\in\cube$ and $i\in[n]=\{1,2,\dots,n\}$. For every vertex $\bx=(x_1,x_2,\ldots,x_n)$ we let $\bar{\bx}=(\bar{x}_1,\bar{x}_2,\ldots,\bar{x}_n)=\bx\oplus(1,1,\dots,1)$.

\begin{definition}[Subcubes and facets]
    If $C=C_n$ and $S=\{(i_1,b_1),\ldots,(i_k,b_k)\}$, where $i_1,\ldots,i_k\in[n]$ are distinct indices and $b_1,b_2,\dots,b_k\in\{0,1\}$, then the \emph{subcube} $C_S=(V_S,E_S)$ of~$C$ is the induced subgraph of~$C$ on the vertex set 
    \[V_S=\{\bx=(x_1,x_2,\ldots,x_n)\in\cube \mid x_{i_1}=b_1, x_{i_2}=b_2, \ldots, x_{i_k}=b_k\}\;.\]
    If $S=\{(i,b)\}$, where $i\in[n]$ and $b\in\{0,1\}$, then $C_S$ is said to be a \emph{facet} of $C$.
\end{definition}

It is not difficult to see that If $S=\{(i_1,b_1),\ldots,(i_k,b_k)\}$, then $C_S$ is isomorphic to $C_{n-k}$. If $S=\emptyset$, then $C_S=C$. A facet $C_{\{(i,b)\}}$ is simply a subcube of~$C=C_n$ of dimension $n-1$. For brevity, we let $C_{i,b}=C_{\{(i,b)\}}$.

\begin{definition}[Orientations, acyclicity and sinks]
    An \emph{orientation} $\vH$ of $C_n$ is obtained by giving each edge of $C_n$ a \emph{direction}. An orientation $\vH$ is thus a directed graph. If $S=\{(i_1,b_1),\ldots,(i_k,b_k)\}$, 
    we let $\vH_S$ be the orientation induced on $C_S$. An orientation $\vH$ is \emph{acyclic} if it does not contain any directed cycles. A \emph{sink} $\bs$ of $\vH$ is a vertex of outdegree~$0$, i.e., all edges incident on~$\bs$ are incoming.
\end{definition}

\begin{definition}[Unique Sink Orientations (USOs)]
    An orientation $\vH$ of $C_n$ is a \emph{Unique Sink Orientation (USO)} if and only if every subcube $\vH_S$ of~$\vH$, including $\vH$ itself, has a \emph{unique} sink.
\end{definition}
\begin{definition}[Acyclic Unique Sink Orientations (AUSOs)]
    An orientation $\vH$ of $C_n$ is an \emph{Acyclic Unique Sink Orientation (AUSO)} if and only if it is a USO and it is also acyclic.
\end{definition}

If a USO, or an AUSO is given explicitly, then its sink can be easily found in linear time.  But, as the size of an explicit representation of an AUSO is exponential in the dimension, this is not an interesting variant of the problem. A much more interesting variant of the problem assumes that we have an oracle access to the AUSO. A \emph{vertex oracle} for an AUSO $\vH$, for example, receives a vertex~$\bx$ and returns the orientation of all the edges incident on~$\bx$. We now get a much more interesting problem: How many \emph{queries} to the oracle are needed to find the sink of a general AUSO?

\section{The Random-Facet algorithm}\label{S-RF}

The Random-Facet algorithm, as adapted to the binary cube by G{\"a}rtner \cite{Gartner02}, has a simple description and an intriguing analysis. We present them both here as they form the basis for our improved algorithm and its analysis.

The algorithm receives an AUSO $H$ of dimension~$n$ and an arbitrary starting point $\bx\in\cube$. (More precisely, the algorithm receives access to an oracle that represents~$H$.) The algorithm chooses, uniformly at random, an index $j\in[n]$ that defines a random facet $\vH_{j,x_j}$ of~$H$. (This gives the algorithm its name.) It recursively finds the sink $\by$ of the facet~$\vH_{j,x_j}$ starting from~$\bx$. If the edge connecting $\by$ and $\bz=\by\oplus \be_j$ is directed from $\bz$ to $\by$, then~$\by$ is also the sink of $\vH$, and~$\by$ is returned. Otherwise, the algorithm moves along the directed edge from $\by$ to $\bz$ and recursively finds the sink of~$\vH$ starting from~$\bz$. (Equivalently, it finds the sink of the facet~$\vH_{j,z_j}$ starting from $\bz$.)

The correctness of the algorithm is obvious. The first recursive call finds the unique sink~$\by$ of $\vH_{j,x_j}$. Note that $y_j=x_j$. If~$\by$ is not also the sink of $\vH$, then the sink of $\vH$ must be the sink of the facet~$\vH_{j,\bar{y}_j}=\vH_{j,z_j}$ which is found by the second recursive call. The algorithm is guaranteed to stop as it can be proved by induction, by unraveling the recursion, that it defines a directed path in~$\vH$.

At first sight, it may seem that the Random-Facet algorithm should have an exponential complexity. To solve a problem of dimension~$n$ it performs two recursive calls on problems of dimension $n-1$. (Worse still, the second recursive call, as first stated, is on a problem of dimension~$n$.) However, a short and surprising analysis, given below, shows that the second recursive call is effectively on a subcube of dimension $n-i$, where~$i$ is chosen uniformly at random from $[n]$. (This $i$ should not be confused with the index~$j$ explicitly chosen by the algorithm.)

Since $\vH$ is acyclic, it has a \emph{topological ordering} $t:\cube\to[2^n]$. If $\bx\to\by$ is a directed edge of~$\vH$, then $t(\bx)<t(\by)$. If $\bs$ is the sink of $\vH$ then $t(\bs)=2^n$. For the analysis of the Random-Facet algorithm it is convenient to fix one such topological ordering.  For every facet $\vH_{i,b}$, let 
\[t(\vH_{i,b}) \EQ \max_{\bx\in \vH_{i,b}} t(\bx) \EQ t(\sink(\vH_{i,b}))\;.\] 
We stress that~$t$ is only used in the analysis and $t(\bx)$, and $t(\vH_{i,x})$, are not known to the algorithm.

Every vertex $\bx=(x_1,\dots,x_n)$ of $\vH$ is on the~$n$ facets $\vH_{1,x_1},\vH_{2,x_2},\ldots,\vH_{n,x_n}$. Let $F_1,F_2,\ldots,F_n$ be the ordering of these facets such that 
\[t(F_1)\LE t(F_2)\LE \dots\LE t(F_n)\;.\]
(Ties are possible as a vertex may be the sink of several facets.) We have $F_i=\vH_{\sigma(i),x_{\sigma(i)}}$ for some permutation~$\sigma$ that is of course not known to the algorithm. Let $\by_i=\sink(F_i)$, for $i\in[n]$.

Suppose that the Random-Facet algorithm chose the facet~$F_i$. The first recursive call then returns $\by_i = \sink(F_i)$. If $\by_i$ is also the sink of the whole cube $\vH$, we are done. Otherwise, the algorithm moves from $\by_i$ to $\bz =\by_i\oplus \be_{\sigma(i)}$. We now have:
\[ t(F_1)\LE t(F_2) \LE \dots \LE t(F_i) \EQ t(\by_i) \LT t(\bz)\;.\]
We can thus infer that $x_{\sigma(j)}\ne z_{\sigma(j)}$, for $j\in [i]$. Furthermore, the bits whose indices are $\sigma(1),\ldots,\allowbreak \sigma(i)$ in the vertices visited from now on by the algorithm can never change! Thus, the second recursive call, if it is made, is effectively on the subcube $\vH_{\{(\sigma(1),z_{\sigma(1)}),\dots,(\sigma(i),z_{\sigma(i)})\}}$ whose dimension is $n-i$ and not~$n$.

Let $F(n)$ be an upper bound on the \emph{expected} number of vertices visited by the algorithm when called on an AUSO of dimension~$n$. This is also an upper bound on the number of vertex queries made by the algorithm. The analysis above implies the following recurrence relation:
\begin{equation}
\begin{aligned}
    F(n) &\EQ F(n-1) + \frac{1}{n} \sum_{i=1}^{n} F(n-i) \quad,\quad n>0 \;,\\[-7pt]
    F(0) &\EQ 1 \;.
    \end{aligned}\label{F-n}
\end{equation}

A nice, and known, solution of this recurrence relation appears in Appendix~\ref{A-basic-recurrence}. It follows from Lemma~\ref{L-f} there that $F(n)\le e^{2\sqrt{n}}$. (Note that $F(n)=f_1(n)$, where $f_1(n)$ is defined in the lemma.)



\section{A fictional algorithm}\label{S-fiction}

As a step towards obtaining an improved version of the Random-Facet algorithm, it is convenient to consider the following \emph{fictional}, i.e., non-realizable, algorithm. 

Suppose that $\bx=(x_1,x_2,\ldots,x_n)$ is a vertex of an AUSO $\vH$. We say that a facet $\vH_{j,b}$, where $j\in[n]$ and $b\in\{0,1\}$, is \emph{$\bx$-active}, i.e., active with respect to~$\bx$, if and only if $t(\bx)\le t(\vH_{j,b})$, or equivalently if there exists a vertex $\bx_{j,b}\in \vH_{j,b}$ such that $t(\bx)\le t(\bx_{j,b})$. We say that $\bx_{j,b}$ \emph{witnesses} the $\bx$-activity of $\vH_{j,b}$. If $\vH_{j,b}$ is not $\bx$-active, then no vertex in $\vH_{j,b}$ can be reached using a directed path from~$\bx$. The notion of active facets was introduced in the context of linear programming by Kalai \cite{Kalai92}.

If $\bx=(x_1,x_2,\dots,x_n)$, then the facets $\vH_{1,x_1},\vH_{2,x_2},\dots,\vH_{n,x_n}$ are automatically $\bx$-active, as we can choose $\bx_{j,b}=\bx$. The Random-Facet algorithm of the previous section begins by choosing one of them at random. Let $\bar{\bx}=(\bar{x}_1,\bar{x}_2,\dots,\bar{x}_n)$. Some of the facets $\vH_{1,\bar{x}_1},\vH_{2,\bar{x}_2},\dots,\vH_{n,\bar{x}_n}$ may also be $\bx$-active. Thus, in general, there might be up to $2n$ active facets with respect to~$\bx$. If $\vH_{j,\bar{x}_j}$ is not $\bx$-active, then any algorithm that only moves, or jumps, from a vertex to a vertex with a higher ranking in the topological ordering would never change the $j$-th coordinate of the vertices visited. In such a case, as in the previous section, the algorithm is effectively working on an AUSO whose dimension is smaller than~$n$.

Suppose that all $2n$ facets $\vH_{j,b}$, where $j\in[n]$ and $b\in\{0,1\}$ are $\bx$-active, and that $\bx_{j,b}\in \vH_{j,b}$, is such that $t(\bx)\le t(\bx_{j,b})$, for every $j\in[n]$ and $b\in\{0,1\}$. (If $b=x_j$, we can let $\bx_{j,b}=\bx$.) Furthermore, suppose that, using some \emph{magical} means, all the witnessing vertices $\bx_{j,b}$ are known to the algorithm. The fictional algorithm now has $2n$ possible choices, instead of only~$n$. It chooses uniformly at random, an active facet $\vH_{j,b}$, where $j\in[n]$, $b\in\{0,1\}$, \emph{jumps}, if needed, from~$\bx$ to~$\bx_{j,b}$, and recursively finds the sink~$\by$ of $\vH_{j,b}$. If~$\by$ is also the sink of~$\vH$ we are done, otherwise, the algorithm moves to $\bz=\by\oplus \be_j$ and recursively finds the sink of~$\vH$ starting from~$\bz$.

Let $F_1,F_2,\ldots,F_{2n}$ be the $2n$ facets of $\vH$, all assumed to be $\bx$-active, arranged such that 
\[t(F_1)\LE t(F_2)\LE \dots\LE t(F_n) \LT t(F_{n+1}) \EQ \dots \EQ t(F_{2n}) \EQ t(\sink(\vH))\;.\]
Indeed, if $\bs=(s_1,s_2,\ldots,s_n)=\sink(\vH)$, then $t(\vH_{i,\bar{s}_i})<t(\vH_{i,s_i})=t(\bs)=2^n$, for $i\in[n]$.

The fictional algorithm chooses, uniformly at random, a facet $F_i$, where $i\in[2n]$, jumps to a vertex~$\bx_i$ on this facet, and recursively finds the sink of the chosen facet. If $n+1\le i\le 2n$, then this recursive call returns the sink~$\bs$ of~$\vH$ and we are done. Otherwise, as in the previous analysis, the first recursive call returns the sink $\by$ of~$F_i$. If $F_i=\vH_{j,b}$, the algorithm moves along a directed edge from~$\by$ to~$\bz=\by\oplus \be_j$, and makes a second recursive call starting from~$\bz$. As in the previous section, this recursive call is effectively on a hypercube of dimension $n-i$. 

This leads to the following modified recurrence relation bounding the complexity of the fictional algorithm:
\begin{equation}
\begin{aligned}
    F(n) &\EQ F(n-1) + \frac{1}{2n} \sum_{i=1}^{n} F(n-i) \quad,\quad n>0 \;,\\[-7pt]
    F(0) &\EQ 1 \;.
    \end{aligned}\label{F-2n}
\end{equation}
The subtle change compared to~\eqref{F-n} is that the sum, which still only goes from~$1$ to~$n$, is now multiplied by $\frac{1}{2n}$ and not $\frac{1}{n}$. Using Lemma~\ref{L-f}, with $\alpha=\frac{1}{2}$, we get that $F(n)\le e^{\sqrt{2n}}$. Thus, the coefficient of $\sqrt{n}$ in the exponent is improved from~$2$ to~$\sqrt{2}$.


\section{The improved algorithm}\label{S-improved}

We next describe a non-fictional improved algorithm. The fictional algorithm suggests that a key feature of an improved algorithm should be the ability to quickly discover many active facets, thus enlarging the pool of active facets from which a random facet is chosen. Following Hansen and Zwick~\cite{HaZw15}, we actually change our perspective and consider the discovery of new active facets to be the main goal of the algorithm. Reaching the sink is, of course, a very desirable side-effect. 

In the case of AUSOs, it turns out to be especially useful to discover complementary pairs of active facets, i.e., pairs of the form $\vH_{j,0}$ and $\vH_{j,1}$, for $j\in[n]$. This is one of the key features of the new algorithm that leads to a much simpler description and analysis compared to the more general algorithm of Hansen and Zwick \cite{HaZw15}.

In the previous section we said that a facet $\vH_{j,b}$ is $\bx$-active if and only if there exists $\bx_{j,b}\in \vH_{j,b}$ such that $t(\bx)\le t(\bx_{j,b})$, where~$t$ is a fixed, but arbitrary, topological ordering of the vertices of~$\vH$. To get an implementable algorithm we need to slightly change this definition.

\begin{definition}[Active facets]
    Let $\vH$ be an AUSO of dimension~$n$ and let $\bx\in\cube$. We say that~$\vH_{j,b}$ is an \emph{$\bx$-active} facet if and only if there is a directed path in~$\vH$ from~$\bx$ to a vertex $\bx_{j,b}\in\vH_{j,b}$.
\end{definition}

If $\vH_{j,b}$ is active according to this definition, then it is also active with respect to any topological ordering~$t$. This definition is more suitable for our purposes as an active facet can now be witnessed by a concrete directed path in~$\vH$ found by the algorithm. (The algorithm does not explicitly maintain the witnessing paths, only their endpoints.)

For every vertex $\bx=(x_1,x_2,\ldots,x_n)$ of~$\vH$ we immediately have $n$ \emph{trivial} active facets $\vH_{i,x_i}$, for $i\in [n]$. Let $\bar{\bx}=(\bar{x}_1,\bar{x}_2,\ldots,\bar{x}_n)=\bx\oplus(1,1,\dots,1)$. An active facet $\vH_{i,\bar{x}_i}$, for some $i\in[n]$, is said to be a \emph{non-trivial} $\bx$-active facet. The facets~$\vH_{i,0}$ and~$\vH_{i,1}$, for some $i\in[n]$ are said to form a \emph{complementary} pair. 

The improved algorithm uses an increasing integral sequence $A=A_c=(a_1,a_2,\ldots)$, with $a_1=1$ and $a_i=\floor{ca_{i-1}+1}$, for $i>1$, for some constant $c>1$. With this definition we have $\frac{a_{i}-1}{c} \le a_{i-1}< a_i$, for $i>1$. The role of the sequence~$A$ and the importance of the condition $\frac{a_{i}-1}{c} \le a_{i-1}< a_i$ will be explained shortly. Note that asymptotically we have $a_i\approx c^i$.
\footnote{Hansen and Zwick \cite{HaZw15} used the simpler sequence $a_i=\ceil{c^i}$ which unfortunately does not satisfy the required condition. This small technical problem can be easily fixed by using the definition of~$A=A_c$ given here.}

If $k>1$, we let 
\[\floor{k}_A = \max\{ a_i \mid a_i<k\,,\, i\ge 1\}\;,\] 
be the largest element in the sequence~$A$ which is strictly smaller than~$k$. Note that if $a_i<k\le a_{i+1}$, then the sequence $\floor{k}_A, \floor{\floor{k}_A}_A,\dots$ is just the sequence $a_i,a_{i-1},\ldots$. For conciseness, we also use the notation $m_k=\floor{k}_A$ when the sequence~$A$ is clear from the context.



The improved algorithm $IRF(\vH,\bx,k)=IRF_A(\vH,\bx,k)$, where~$A$ is the sequence used, receives an AUSO~$\vH$ of dimension~$n$, an initial vertex $\bx\in \cube$ of~$\vH$, and the number~$k$ of complementary pairs of $\bx$-active facets that the algorithm should find. The algorithm either returns the sink of~$\vH$, or a list of $2k$ distinct $\bx$-active facets $F_1,F_2,\ldots,F_{2k}$, forming~$k$ complementary pairs, along with witnesses $\bx_1,\bx_2,\ldots,\bx_{2k}$, where $\bx_i\in F_i$ is a vertex that can be reached from~$\bx$ by a directed path in~$\vH$.
As there are at most~$2n$ facets, and in particular at most~$2n$ $\bx$-active facets, $IRF(\vH,\bx,n+1)$ must return the sink of~$\vH$.

At a high level, the improved algorithm works as follows. If $k=1$, the algorithm checks whether $\bx$ has any outgoing edges. \footnote{This is the only place in which the algorithm uses a vertex query. All other queries made by the algorithm are edge queries. For simplicity, we count all queries as vertex queries.} If it does not have such edges, then~$\bx$ is the sink of~$H$ and the algorithm returns it. Otherwise, if $\bx\to\bx\oplus \be_i$ is an outgoing edge, the algorithm returns the complementary pair $H_{i,x_i}$ and $H_{i,\bar{x}_i}$ of $\bx$-active facets with the witnesses $\bx$ and $\bx\oplus\be_i$. 

Suppose now that $k>1$. As a first step towards finding~$k$ complementary pairs of $\bx$-active facets, or finding the sink of~$\vH$, the algorithm sets an intermediate goal of finding $m_k=\floor{k}_A<k$ such pairs, or finding the sink of~$H$. This intermediate goal is obtained by a recursive call $IRF(\vH,\bx,m_k)$. This first recursive call either returns the sink of~$\vH$, or a list of $2m_k$ distinct $\bx$-active facets $F_1,F_2,\ldots,F_{2m_k}$, composed of $m_k$ complementary pairs, along with corresponding witnesses $\bx_1,\bx_2,\dots,\bx_{2m_k}$. 

If the first recursive call finds the sink of~$\vH$, then we are of course done. Otherwise, the algorithm chooses a random $\bx$-active facet $F_j$, where~$j$ is chosen uniformly at random from $[2m_k]$. It jumps, if needed, from~$\bx$ to~$\bx_j$ and performs a second recursive call $IRF(F_j,\bx_j,k-1)$ that either finds~$k-1$ complementary pairs of $\bx_j$-active facets $G_1,G_2,\ldots,G_{2(k-1)}$ of~$F_j$, or finds the sink of~$F_j$.

If $k-1$ complementary pairs $G_1,G_2,\ldots,G_{2(k-1)}$ of $\bx_j$-active facets of~$F_j$ are returned by the second recursive call, then $G_1,G_2,\ldots,\allowbreak G_{2(k-1)}$ extended to facets of~$\vH$ by adding the coordinate fixed by~$F_j$, along with the pair containing~$F_j$, form a list of $k$ complementary pairs of $\bx$-active facets that can be returned by the algorithm. (More specifically, if $F_j=H_{f_j,b_j}$, where $f_j\in[n]$, $b_j\in\{0,1\}$, and $G_1,G_2,\ldots,G_{2(k-1)}$ are determined  by the distinct coordinates $g_1,\dots,g_{k-1}\in[n]$, all different from~$f_j$ i.e., $G_{2i-1}=(F_j)_{g_i,0}, G_{2i}=(F_j)_{g_i,1}$, the call can return the complementary pairs determined by coordinates $f_j,g_1,\dots,g_{k-1}$, along with appropriate witnesses.)

Otherwise, the second recursive call finds the sink~$\by$ of~$F_j$ before finding enough pairs of $\bx_j$-active facets. If~$\by$ is also the sink of~$\vH$, we are again done. If not, the algorithm moves along an edge from~$\by$ to~$\bz=\by\oplus \be_{f_j}$, where~$f_j$ is the coordinate fixed by~$F_j$. It then makes a third, and last, recursive call $IRF(\vH,\bz,k)$ that either finds~$k$ complementary pairs of $\bz$-active facets, or finds the sink of~$\vH$. Every $\bz$-active facet is clearly also an $\bx$-active facet, as there is a directed path from~$\bx$ to~$\bz$ in~$\vH$.

The third recursive call $IRF(\vH,\bz,k)$, as stated above, is asked to~find~$k$ complementary pairs of $\bz$-active facets, which as explained are also pairs of $\bx$-active facets. This, however, may be more than necessary, as the first two recursive calls already found some complementary pairs of $\bx$-active facets. In particular, the first recursive call already found $m_k=\floor{k}_A$ of them. Is it then enough to ask the third recursive call to find only $k-m_k$ complementary pairs of $\bz$-active facets? Unfortunately, the answer is no, as the third recursive call may find again some of the pairs already found by the first or second recursive calls. It would be enough, however, to find $k-m_k$ \emph{new} active pairs, or more generally, to complete the collection of complementary $\bx$-active pairs already found by the first two recursive calls to a collection of $k$ distinct $\bx$-active pairs.

In general, we have no guarantee that pairs of active facets found by the first recursive call will not be found again by the second or third recursive calls. However, the analysis below shows that we \emph{do} have a guarantee that \emph{some} active pairs found by the second recursive call will \emph{not} be found again by the third recursive call. This is crucial for the efficiency of the algorithm. We do not know in advance to how many pairs this guarantee applies.

To make sure that all active pairs of facets already discovered by the algorithm are used, we add the following details to the implementation of the algorithm. Each recursive call $IRF(H,\bx,k)$ keeps a collection~$\cF$ of complementary $\bx$-active pairs of facets that were already discovered, either by this recursive call, or by any descendant recursive call. Initially $\cF=\emptyset$. (If $\bx$ has~$i$ outgoing edges, we can initialize~$\cF$ to the~$i$ corresponding pairs of $\bx$-active facets. This is actually used at the bottom of the recursion when $k=1$.)
When the algorithm traverses an edge $\bx\to \bx\oplus \be_i$, it actually `discovers' the complementary pair $H_{i,0}$ and $H_{i,1}$ of active facets. (As mentioned, this pair may have already been discovered before.)
When a pair of active facets is discovered, a recursive call tries to add it not only to its own collection, but also to the collection of its parent, and in fact to the collection of each one of its ancestors. Each collection has a target size. When this target size is reached, all recursive calls below its level can be aborted.

More precisely, a call $IRF(\vH,\bx,k)$ generates a chain of active recursive calls $IRF(\vH_i,\bx_i,k_i)$, for $i=0,1,\ldots,d$, where~$d$ is the current depth of the recursion. We have $\vH_0=\vH$, $\bx_0=\bx$ and $k_0=k$. For every~$0<i\le d$, either $\vH_i=\vH_{i-1}$ or $\vH_i$ is a facet of $\vH_{i-1}$. Thus, $\vH_j$ is a subcube of $\vH_i$, for every $0\le i<j\le d$. The algorithm maintains the invariant that there is a directed path in~$\vH$ that passes through $\bx_0,\bx_1,\ldots,\bx_d$, in this order. (The algorithm does not need to explicitly maintain this path.) Thus, an $\bx_j$-active facet immediately gives rise to an $\bx_i$-active facet for every $i<j$. (We simply need to add back the coordinates, if any, that were fixed by the recursive calls when moving from~$\vH_i$ to~$\vH_j$.) 

The recursive call $IRF(\vH_i,\bx_i,k_i)$ at depth~$i$ keeps a collection $\cF_i$ of the pairs of $\bx_i$-active facets already discovered and the size of this collection. Initially $\cF_i$ is empty. When the recursive call at depth~$d$ discovers a new pair of $\bx_d$-active facets, this pair is added not only to~$\cF_d$, but also to $\cF_0,\cF_1,\dots,\cF_{d-1}$. (Note that new pair is always discovered by the deepest recursive call.) It is possible that the pair just discovered already appeared in~$\cF_i$, for some $i<d$, in which case it is not added again, and the size of $\cF_i$ remains unchanged.

When $\cF_i$ reaches the target number~$k_i$ of complementary pairs of $\bx_i$-active facets, the goal of the recursive call at depth~$i$ is achieved. The recursive calls at depths $i+1,i+2,\ldots,d$ can be aborted and the recursive call at depth~$i$ can return the collection~$\cF_i$ to its parent, or to the outside caller. 

To implement this nonlocal return mechanism, we let each recursive call maintain its depth and assume that a recursive call at depth~$i$ has access to the collections $\cF_0,\cF_1,\ldots,\cF_i$. A recursive call at depth~$i$ either returns an index~$j\le i$ that indicates that the goal of the depth~$j$ recursive call was achieved, and that this is the smallest such index, or a vertex~$\bz$ which is the sink of the subcube $\vH_i$ on which it was called. If a recursive call at depth~$i$ receives an index $j<i$ as the result of one of the recursive calls it made, it immediately passes this index to its parent. In such a way, the `signal' propagates all the way to the required level.

With these additional details, the correctness of the algorithm is essentially self-evident. We next turn our attention to the main challenge of analyzing the algorithm. 


The crucial point of the analysis is the following. Assume, for the sake of the analysis, that the $2m_k$ facets $F_1,F_2,\ldots,F_{2m_k}$, where $m_k=\floor{k}_A$, composed of $m_k$ complementary pairs, returned by the first recursive call, if it did not find the sink of~$\vH$, are ordered such that 

\[t(F_1)\LE t(F_2)\LE \dots\LE t(F_{m_k}) \LT t(F_{{m_k}+1}) \EQ \dots \EQ t(F_{2{m_k}}) \EQ t(\sink(\vH))\;,\]
where $t$ is a fixed, but arbitrary, topological ordering of~$\vH$ used only in the analysis.
Let $\by_i$ be the sink of~$F_i$, for $i\in[2{m_k}]$. If the second recursive call is on~$F_i$, where~$i$ is chosen uniformly at random, and the second recursive call finds the sink~$\by_i$ of~$F_i$, and~$\by_i$ is not the sink of~$\vH$, then the algorithm moves from $\by_i$ to $\bz=\by\oplus \be_{f_i}$, and as in the previous sections we have
\[ t(F_1)\LE t(F_2) \LE \dots \LE t(F_i) \EQ t(\by_i) \LT t(\bz)\;.\]
That means that $\bz$, and all vertices that can be reached from it, do not belong to $F_1,F_2,\ldots,F_i$. Furthermore, if $f_1,f_2,\dots,f_i$ are the coordinates defining $F_1,F_2,\ldots,F_i$, then these coordinates are all frozen. To meet the goal of discovering $k$ complementary pairs of $\bx$-active facets, it is thus enough to find $k-i$ pairs of $\bz$-active facets, as these pairs are guaranteed to be different from the pairs containing $F_1,F_2,\ldots,F_i$. Thus, without the algorithm knowing it, the third recursive call is effectively of the form $IRF(\vH',\bz,k-i)$, where $\vH'$ is an $(n-i)$-dimensional subcube of~$\vH$. The analysis of the algorithm relies crucially on this point.

There is actually one difference between the third recursive call $IRF(H,\bz,k)$ and the `effective' recursive call $IRF(H',\bz,k-i)$ that we would like to use in the analysis, even if the third recursive call does effectively work on~$H'$, an $(n-i)$-dimensional subcube of~$\vH$, and will indeed terminate when finding $k-i$ pairs of active facets. The actual recursive call sets an intermediate target of finding $m_k=\floor{k}_A$ active pairs, while the target set by $IRF(H',\bz,k-i)$ is only $m_{k-i}=\floor{k-i}_A$, which might be smaller. This is exactly the reason for using the fixed sequence~$A$. Suppose that $\floor{k}_A=a_j$ and that $\floor{k-i}_A=a_{j'}$, for some $j'\le j$. Assume that $j'<j$. The actual recursive call $IRF(H,\bz,k)$ immediately initiates a sequence of recursive calls $IRF(H,\bz,a_j)$, $IRF(H,\bz,a_{j-1})$, \dots until reaching $IRF(H,\bz,a_{j'})$. When $IRF(H,\bz,a_{j'})$ terminates, it returns control to $IRF(H,\bz,a_{j'+1})$. Since $a_{j'+1}\ge k-i$, $IRF(H,\bz,a_{j'+1})$ is effectively equivalent to $IRF(\vH',\bz,k-i)$. (Both will either find the sink of the cube or will terminate after finding at most $k-i$ active pairs of facets.) The time spent in the other recursive calls of the chain is negligible and does not involve any oracle calls.

Let $F(k)=F_A(k)$ be an upper bound on the expected number of vertices visited by the algorithm until either finding $k$ complementary pairs of active facets with respect to the starting vertex, or finding the sink of the input AUSO. (The improved algorithm, unlike the original algorithm, may return a vertex several times. Each visit to a vertex is counted separately.) As mentioned, the number of vertices visited is also a bound on the number of vertex queries made by the algorithm. An interesting feature of the analysis given here is that the upper bound~$F(k)$ depends only on~$k$, and not on the dimension~$n$ of the input hypercube. This makes the analysis much simpler than the analyses given in \cite{HaZw15} and \cite{Zwick26}. 

By the above discussion, we get the following recurrence relation for $F(k)=F_A(k)$, where $m_k=\floor{k}_A$:

\begin{equation}
\begin{aligned}
    F(k) &\EQ F(m_k) + F(k-1) + \frac{1}{2m_k} \sum_{i=1}^{m_k} F(k-i) \quad,\quad k>1 \;,\\[-5pt]
    F(1) &\EQ 1 \;.
    \end{aligned}\label{F-eq}
\end{equation}

This recurrence relation is similar to the recurrence relations~\eqref{F-n} and~\eqref{F-2n} appearing in the previous two sections, and in Lemma~\ref{L-f}, with two main differences. The first is an additional term $F(m_k)$. The second is that the last term involves $m_k$ instead of~$k$. Another minor difference is that the basis of the recurrence is now $k=1$ instead of $k=0$.

In the next section we solve this recurrence relation and also explain how the constant~$c$ that appears in the definition of the sequence~$A$ is chosen. With the right choice of~$c$, we show that the term $F(m_k)$ is negligible, and that the average $\frac{1}{2m_k} \sum_{i=1}^{m_k} F(k-i)$ is not far from the average $\frac{1}{2k} \sum_{i=1}^{k} F(k-i)$ appearing in~\eqref{F-2n}. Thus, the solution of~\eqref{F-eq} is not much larger than $e^{\sqrt{2k}}$.

 \section{Solving the recurrence relation}\label{S-solve}

Throughout this section, we let $F(k)=F_A(k)$ be the function defined in Equation~\eqref{F-eq}. We now define two new functions $G(k)=G_A(k)$ and $H(k)=H_A(k)$ as follows, where again $m_k=\floor{k}_A$:
\begin{align}
  &\begin{aligned}
    G(k) &\EQ G(k-1) + \frac{1}{2m_k}
      \sum_{i=1}^{m_k} G(k-i)
      \quad,\quad k>1 \;,\\[-8pt]
    G(1) &\EQ 1 \;.
  \end{aligned}
  \label{G}\\[8pt]
  &\begin{aligned}
    H(k) &\EQ H(m_k) + H(k-1)
      \quad,\quad k>1 \;,\\
    H(1) &\EQ 1 \;.
  \end{aligned}
  \label{H}
\end{align}

Note that $G(k)$ includes the last two terms in the definition of~$F(k)$ while $H(k)$ includes the first two terms in the definition of~$F(k)$. Since the definitions of $G(k)$ and $H(k)$ contain the terms $G(k-1)$ and $H(k-1)$, and since all terms appearing in the definitions are nonnegative, we immediately get that $G(k)$ and $H(k)$ are monotonically increasing functions of~$k$. The following lemma is similar to lemmas appearing in G{\"a}rtner \cite{Gartner95} and Hansen and Zwick \cite{HaZw15}.

\begin{lemma}\label{L-FGH}
    We have $F(k)\le G(k)H(k)$, for every $k\ge 1$.
\end{lemma}

\begin{proof}
    The proof is by an easy induction on~$k$. For $k=1$, we have $F(1)=G(1)=H(1)=1$. Suppose that the claim holds for all values less than~$k$. Then, with $m_k=\floor{k}_A$ we have:
    \begin{alignat*}{2}
  F(k)
  &\EQ F(m_k) + F(k-1) + \frac{1}{2m_k}\sum_{i=1}^{m_k}F(k-i) 
  && \quad \text{(definition of $F$)} \\
  &\LE G(m_k)H(m_k) + G(k-1)H(k-1) + \frac{1}{2m_k}\sum_{i=1}^{m_k}G(k-i)H(k-i)
  && \quad \text{(induction hypothesis)} \\
  &\LE G(k-1)(H(m_k)+H(k-1)) + \left(\frac{1}{2m_k}\sum_{i=1}^{m_k}G(k-i)\right)H(k)
  && \quad
     \text{($G$ and $H$ are monotone)} \\
  &\EQ G(k-1)H(k) + (G(k)-G(k-1))H(k)
  && \quad \text{(definitions of $G$ and $H$)} \\
  &\EQ G(k)H(k)
  && \tag*{\qedhere}
\end{alignat*}
\end{proof}

It is thus enough to bound each of the functions $G(k)$ and $H(k)$ separately. To bound $G(k)$, we use the function $g(k)=g_\alpha(k)$, where $\alpha>0$, defined as:
\begin{equation}
\begin{aligned}
    g(k) &\EQ g(k-1) + \frac{\alpha}{k-1}\sum_{i=1}^{k-1} g(k-i) \quad,\quad k>1 \;,\\[-6pt]
    g(1) &\EQ 1 \;.
    \end{aligned}\label{g}
\end{equation}
Lemma~\ref{L-g} of Appendix~\ref{A-basic-recurrence} shows that $g_\alpha(k)$ is a very close relative of the recurrence relations defined in Equations~\eqref{F-n} and \eqref{F-2n} and that essentially the same analysis shows that $g_\alpha(k)< e^{2\sqrt{\alpha (k-1)}}$.

\begin{lemma}\label{L-G}
    If $c>1$ and $k > \floor{k}_A\ge \frac{k-1}{c}$, for $k>1$, then $G_A(k)\le g_{c/2}(k)$, where $g_{c/2}(k)$ is the function defined in~\eqref{g} and in Lemma~\ref{L-g}. Thus $G_A(k) < e^{2\sqrt{\frac{c}{2}(k-1)}} = e^{\sqrt{2c(k-1)}}$, for $k\ge 1$.
\end{lemma}

\begin{proof}
    Let $G(k)=G_A(k)$ and $g(k)=g_{c/2}(k)$. We prove the inequality by induction on~$k$. As the basis of the induction we have $G(1)=g(1)=1$. We again let $m_k=\floor{k}_A$. For $k>1$ we then have:
    \begin{alignat*}{2}
G(k)
  &\EQ G(k-1)
     + \frac{1}{2m_k}
       \sum_{i=1}^{m_k} G(k-i)
  && \qquad \text{(definition of \(G\))} \\
  &\LE G(k-1)
     + \frac{1}{2m_k}
       \sum_{i=1}^{k-1} G(k-i)
  && \qquad \text{(summing more terms)} \\
  &\LE G(k-1)
     + \frac{c}{2(k-1)}
       \sum_{i=1}^{k-1} G(k-i)
  && \qquad
     \text{(since $\textstyle m_k=\floor{k}_A \ge \frac{k-1}{c}$)} \\
  &\LE g(k-1)
     + \frac{c}{2(k-1)}
       \sum_{i=1}^{k-1} g(k-i)
  && \qquad \text{(induction hypothesis)} \\
  &\EQ g(k)
  && \qquad \text{(definition of \(g\)).} \tag*{\qedhere}
\end{alignat*}

\end{proof}

We next move to the analysis of $H(k)=H_A(k)$. The analysis give here is similar to an analysis given in Hansen and Zwick \cite{HaZw15} for a slightly different function.

\begin{lemma}\label{L-H-closed}
    Let $A=(a_1,a_2,\dots)$ be an increasing integral sequence with $a_1=1$. Then, the function $H(k)=H_A(k)$ satisfies
    {
    \addtolength{\leftmargini}{0.5em}    
    \begin{enumerate}
    \renewcommand{\labelenumi}{\normalfont(\roman{enumi})}
    \addtolength{\labelsep}{0.3em}
        \item $H(k) \EQ j\, H(a_{i-1}) + H(k-j) \quad,\quad a_{i-1}< k\le a_i \;,\; 0\le j\le k-a_{i-1}$
        \item $H(a_i) \EQ (a_i-a_{i-1}+1)\, H(a_{i-1}) \quad,\quad i>1$
        \item $\displaystyle H(a_i) \EQ \prod_{j=1}^{i-1} (a_{j+1}-a_{j}+1) \quad,\quad i>1$
    \end{enumerate}
    }
\end{lemma}

\begin{proof}
    We start with (i). For $j=1$, by the definition of~$H$ we have 
    \[H(k)=H(m_k)+H(k-1) \EQ H(a_{i-1})+H(k-1) \;.\] 
    By induction we then have 
    \begin{align*}
        H(k) &\EQ (j-1)H(a_{i-1}) + H(k-j+1)\\
        &\EQ (j-1)H(a_{i-1}) + (H(a_{i-1}) + H(k-j)) \EQ j\, H(a_{i-1}) + H(k-j) \;,
    \end{align*}
    proving (i). (ii) is obtained from~(i) using $k=a_i$ and $j=a_i-a_{i-1}$. Finally, (iii) is obtained from~(ii) by induction, recalling that $H(1)=1$.
\end{proof}

Recall that for a given $c>1$, we defined the sequence $A=A_c=(a_1,a_2,\dots)$ as follows: $a_1=1$ and $a_i=\floor{ca_{i-1}}+1$, for $i>1$. We clearly have $a_i\le ca_{i-1}+1$ and it follows easily by induction that $a_i\le\frac{c^i-1}{c-1}$, for $i\ge 1$. (Indeed, $c\frac{c^{i-1}-1}{c-1}+1 = \frac{c^i-1}{c-1}$.)

\begin{lemma}\label{L-H-ai}
    If $A=A_c$ is as defined above, then 
    \[ H(a_i) \LE e^{\frac{1}{c-1}}c^{\binom{i}{2}} \quad,\quad i\ge 1\;.\]
\end{lemma}

\begin{proof}
    Since $a_i\le\frac{c^i-1}{c-1}$,  we have
    \[ a_{j+1}-a_{j}+1 \LE (c-1)a_{j}+2 \LE (c-1)\frac{c^{j}-1}{c-1}+2 \EQ c^{j}+1\;.\]
    By Lemma~\ref{L-H-closed}(iii), and using $1+x\le e^{x}$ and $\sum_{j=1}^\infty x^{-j}=\frac{1}{x-1}$, we have
    \begin{align*}
        H(a_i) &\EQ \prod_{j=1}^{i-1} (a_{j+1}-a_{j}+1) \LE \prod_{j=1}^{i-1}(c^j+1) \EQ c^{\sum_{j=1}^{i-1}j} \prod_{j=1}^{i-1} (1+c^{-j}) \\
        &\LE c^{\binom{i}{2}} \prod_{j=1}^{i-1}e^{c^{-j}} \EQ c^{\binom{i}{2}} e^{\sum_{j=1}^{i-1}c^{-j}} \LE c^{\binom{i}{2}} e^{\sum_{j=1}^{\infty}c^{-j}} \EQ c^{\binom{i}{2}}e^{\frac{1}{c-1}} \;. \qedhere
    \end{align*}
\end{proof}

\begin{corollary}\label{C-H}
    If $A=A_c$ is as defined above, then 
    \[ H(k) \LE c\,e^{\frac{1}{c-1}}\,k^{3/2}\, e^{\frac{\log^2k}{2\log c}} \quad,\quad k\ge 1\;.\]
\end{corollary}

\begin{proof}
    If $a_{i-1}< k\le a_i$, then $H(k)\le H(a_i)$. It follows easily by induction that $c^{i-1}\le a_i$. Thus $c^{i-2}\le a_{i-1}<k$ and $i< \log_c k + 2$. By Lemma~\ref{L-H-ai} we thus have:
    \[ H(k) \LE e^{\frac{1}{c-1}}c^{\binom{i}{2}} \LE e^{\frac{1}{c-1}}c^{\frac{1}{2}(\log_c k +2)(\log_c k +1)} \EQ e^{\frac{1}{c-1}}c^{\frac12 (\log_c k)^2} c^{\frac32 \log_c k}c \EQ c\,e^{\frac{1}{c-1}}\,k^{3/2}\, e^{\frac{\log^2k}{2\log c}} \;. \qedhere \]
\end{proof}

\begin{corollary}\label{C-2cn}
        For every $c>1$, the expected number of vertices visited by the improved algorithm that uses the sequence $A=A_c$ on an AUSO of dimension~$n$ is at most $e^{\sqrt{2cn}+O_c(\log^2 n)}$.
\end{corollary}

\begin{proof}
    The expected number of vertices visited by the algorithm is at most $F(n+1)$. By Lemma~\ref{L-FGH} we have $F(n+1)\le G(n+1)H(n+1)$. By Lemma~\ref{L-G} and Corollary~\ref{C-H} we thus have
    \[ F(n+1) \LE G(n+1)H(n+1) \LE e^{\sqrt{2cn}} \, c\,e^{\frac{1}{c-1}}\,(n+1)^{3/2}\, e^{\frac{\log^2(n+1)}{2\log c}} \EQ e^{\sqrt{2cn}+O_c(\log^2 n)} \;.\]
    Note that for a fixed $c>1$, the contribution of $H(n+1)$ is negligible compared to the contribution of~$G(n+1)$. (The notation $O_c(\log^2 n)$ means that the constant factor hidden by the big-$O$ notation depends on~$c$.)
\end{proof}

By choosing $c=1+\varepsilon$, for an arbitrarily small $\varepsilon>0$, we can thus get an algorithm whose complexity is $e^{\sqrt{2(1+\varepsilon)n}+O_\varepsilon(\log^2 n)}$. We can do even better by choosing~$c$ to be a function of~$n$. (Note that we fix~$c$ and the sequence $A=A_C$ at the start, depending on the dimension~$n$ of the AUSO given, and keep the same sequence throughout all the recursive calls.)

\begin{theorem}\label{T-main}
    The expected number of vertices visited by the improved algorithm that uses the sequence $A=A_c$, where $c=1+\frac{\log n}{(2n)^{1/4}}$, on an AUSO of dimension~$n$ is at most 
    \[ e^{\sqrt{2n} + (2n)^{1/4}\log n +O(n^{1/4}/\log n)} \;.\]
\end{theorem}

\begin{proof}
    As in the proof of Corollary~\ref{C-2cn} we have 
    \[F(n+1) \le e^{\sqrt{2cn}} \, c\,e^{\frac{1}{c-1}}\,(n+1)^{3/2}\, e^{\frac{\log^2(n+1)}{2\log c}} \;.\]
    We next bound each one of the terms appearing in this expression when $c=1+\frac{\log n}{(2n)^{1/4}}$. The dominant term is of course:
    \[ e^{\sqrt{2cn}} \EQ e^{\sqrt{2n}\sqrt{1+\frac{\log n}{(2n)^{1/4}}}} \LE e^{\sqrt{2n}\left(1+\frac12\frac{\log n}{(2n)^{1/4}}\right)} \EQ e^{\sqrt{2n} + \frac12 (2n)^{1/4}\log n}\;,\]
    where we used the inequality $\sqrt{1+x}\le 1+\frac{x}{2}$, for $x\ge 0$. We next have
    \[ e^{\frac{\log^2(n+1)}{2\log c}} \LE e^{\frac12\log^2(n+1)\left(\frac{(2n)^{1/4}}{\log n}+\frac12\right)} \EQ e^{\frac12 (2n)^{1/4}\log n + O(\log^2n)} \;,\]
    where we used the inequalities $\frac{1}{\log c}\le \frac{1}{c-1}+\frac{1}{2}$ and $\log(n+1)=\log n + \log(1+\frac1n)\le \log n + \frac1n$. Note that $c=1+\frac{\log n}{(2n)^{1/4}}$ was carefully chosen to balance the second order term of $e^{\sqrt{2cn}}$ with the dominant term of $e^{\frac{\log^2(n+1)}{2\log c}}$. Finally,
    \[ e^{\frac{1}{c-1}} \EQ e^{(2n)^{1/4} /\log n} \quad,\quad c(n+1)^{3/2}\EQ e^{O(\log n)}\;. \qedhere \]
\end{proof} 

\section{A variant of the improved algorithm and its analysis}\label{S-tighter}

In this section we obtain a variant of the improved algorithm that uses a fixed sequence~$A$ that does not need to be tuned depending on the dimension~$n$ of the input AUSO. We also obtain a very simple self-contained analysis of this algorithm that does not rely on the results of Appendix~\ref{A-basic-recurrence}.

The sequence used is particularly simple: $a_i=i^4$, for $i\ge 1$. This sequence is chosen as it mimics the asymptotic behavior of the sequence $A_c$, where $c=1+\frac{\log n}{(2n)^{1/4}}$.

We begin with a simple lemma, with a surprisingly simple proof, that bounds the function $G(k)=G_A(k)$ for any increasing sequence~$A$. This lemma can serve as a substitute to the lemmas of Appendix~\ref{A-basic-recurrence}.

\begin{lemma}\label{L-G-bound}
    Let $A=(a_1,a_2,\dots)$ be a fixed increasing sequence with $a_1=1$ and let $m_k=\floor{k}_A$, for $k>1$. Let $G(k)=G_A(k)$ be the function as defined in Equation~\eqref{G} above:
\begin{equation*}
\begin{aligned}
    G(k) &\EQ G(k-1) + \frac{1}{2m_k} \sum_{i=1}^{m_k} G(k-i) \quad,\quad k>1 \;,\\[-8pt]
    G(1) &\EQ 1 \;.\\[-8pt]    
    \end{aligned}
    \end{equation*}
    Then,
    \[ G(k) \LE \exp\left(\sum_{t=2}^k \frac{1}{\sqrt{2m_t}}\right) \quad,\quad k\ge 1\;.\]
\end{lemma}

\begin{proof}
    Let 
    \[ \delta_k\EQ\frac{1}{\sqrt{2m_k}} \quad,\quad S_k \EQ \sum_{t=2}^k \delta_t \quad,\quad B_k \EQ e^{S_k} \quad,\quad k\ge 1 \;.\]
    Since $m_k$ is non-decreasing in~$k$, $\delta_k$ is non-increasing in~$k$ and thus
    \[ S_k-S_{k-i} \EQ \sum_{t=k-i+1}^k \delta_t \GE i \delta_k \;.\]
    Thus,
    \[ B_{k-i} \EQ B_k\, e^{-(S_k-S_{k-i})} \LE B_k\, e^{-i\delta_k} \;,\]
    and 
\begin{equation}\label{eq:B-bound}
\begin{aligned}
    \frac{1}{2m_k}\sum_{i=1}^{m_k} B_{k-i}
    &\LE \frac{B_k}{2m_k}\sum_{i=1}^{m_k}e^{-i\delta_k}
    \LE \frac{B_k}{2m_k}\sum_{i=1}^{\infty}e^{-i\delta_k}
    \EQ \frac{B_k}{2m_k(e^{\delta_k}-1)}
    \\
    &\EQ B_k\frac{\delta_k^2}{e^{\delta_k}-1}
    \LE B_k(1-e^{-\delta_k})
    \EQ B_k-B_{k-1} \;.
\end{aligned}
\end{equation}
    where we used the facts that $\sum_{i=1}^\infty x^{-i} =\frac{1}{x-1}$, $\frac{1}{2m_k}=\delta_k^2$, that $\frac{\delta^2}{e^\delta-1}\le 1-e^{-\delta}$, for every $\delta>0$, and finally that $B_ke^{-\delta_k}=B_{k-1}$.
    
    We prove by induction on~$k$ that $G(k)\le B_k$, thus proving the lemma. As the basis of the induction we have $G(1)=B_1=1$. As the induction step we have:
        \begin{alignat*}{2}
G(k)
  &\EQ G(k-1)
     \;+\; \frac{1}{2m_k}\sum_{i=1}^{m_k} G(k-i)
  && \qquad \text{(definition of \(G\))} \\
  &\LE B_{k-1}
     \;+\; \frac{1}{2m_k}
       \sum_{i=1}^{m_k} B_{k-i}
  && \qquad \text{(induction hypothesis)} \\
  &\LE B_{k-1} \;+\; (B_k-B_{k-1})
  && \qquad
     \text{(by \eqref{eq:B-bound}}) \\
  &\EQ B_k
  && \tag*{\qedhere}
\end{alignat*}
\end{proof}

We now choose the sequence $a_i=i^4$. We note that $m_k=\floor{k}_A=\floor{(k-1)^{1/4}}^4$.

\begin{lemma}\label{L-G4}
    With the choice $a_i=i^4$ we have $G(k)= e^{\sqrt{2k}+O(k^{1/4})}$.
\end{lemma}

\begin{proof}
    If $j^4<t\le (j+1)^4$, then $m_t=j^4$. By Lemma~\ref{L-G-bound} we have $G(k)\le\exp(\sum_{t=2}^k\frac{1}{\sqrt{2m_t}})$. Thus, we need to bound $S_k=\sum_{t=2}^k\frac{1}{\sqrt{2m_t}}$.

    Let $q=\floor{(k-1)^{1/4}}$, so that $q^4<k\le (q+1)^4$, and let 
    \[L_j \EQ a_{j+1}-a_j \EQ (j+1)^4-j^4 \EQ 4j^3+6j^2+4j+1 \LE 4j^3+11j^2\quad,\quad j\ge 1\;,\] be the number of values of~$t$ for which $m_t=j^4$. 
    \[ S_k  \EQ \sum_{t=2}^k\frac{1}{\sqrt{2m_t}} \LE \sum_{j=1}^q \frac{L_j}{\sqrt{2j^4}} \LE \frac{4}{\sqrt{2}}\sum_{j=1}^q j + \frac{11}{\sqrt{2}}\sum_{j=1}^q 1 \EQ \frac{4}{\sqrt{2}}\frac{q(q+1)}{2} + \frac{11}{\sqrt{2}}q \EQ \sqrt{2}q^2 + \frac{13}{\sqrt{2}}q\;. \]
    The first inequality follows as $S_k$ includes exactly $L_j$ terms equal to $\frac{1}{\sqrt{2m_j}}$, for $1\le j<q$, and $k-q^4\le L_q$ terms equal to $\frac{1}{\sqrt{2m_q}}$.

    Since $q=\floor{(k-1)^{1/4}}< k^{1/4}$, we have
    \begin{align*}
        S_k &\LE \sqrt{2}q^2 + \frac{13}{\sqrt{2}}q \LT \sqrt{2k} + \frac{13}{\sqrt{2}}k^{1/4} \;.
    \end{align*}
    The constant factor multiplying $k^{1/4}$ can be improved using a slightly more careful analysis.
\end{proof}

\begin{lemma}\label{L-H4}
    With the choice $a_i=i^4$, we have $H(k)\le e^{\frac{3}{4}k^{1/4}\log k + O(k^{1/4})}$.
\end{lemma}

\begin{proof}
    As in the previous lemma, let $q=\floor{(k-1)^{1/4}}$ and $L_j = (j+1)^4-j^4 = 4j^3+6j^2+4j+1$. We this time use the bound $L_j\le 4(j+1)^3-1$. By Lemma~\ref{L-H-closed}, we have
    \[ H(q^4) \EQ \prod_{j=1}^{q-1}(L_j+1) \LE \prod_{j=1}^{q-1}4(j+1)^3 \EQ 4^{q-1}(q!)^3\;,\]
    and
    \[ H(k) \LE H((q+1)^4) \EQ \prod_{j=1}^{q}(L_j+1) \LE 4^{q}((q+1)!)^3\;.\]    
    Using the trivial bound $(q+1)!\le (q+1)^{q+1}$, and the fact that $q\le (k-1)^{1/4}$, we get that
    \begin{align*}
        \log H(k) &\LE q\log 4 + 3(q+1)\log(q+1)\\
        &\LE (k-1)^{1/4}\log 4 + 3((k-1)^{1/4}+1)\log((k-1)^{1/4}+1) \\
        &\EQ \frac{3}{4}k^{1/4}\log k + O(k^{1/4}) \;,
    \end{align*}
    and the lemma follows. We can further show that
    \[ q\log 4 + 3(q+1)\log(q+1) \LE 6q + 3 q\log q \quad,\quad q\ge 1\;.\]
    This implies the explicit upper bound
    $ \log H(k) \LE \frac{3}{4}k^{1/4}\log k + 6k^{1/4} \;.$ \qedhere 
\end{proof}

By combining Lemmas~\ref{L-G4} and~\ref{L-H4}, together with Lemma~\ref{L-FGH}, we get:
    \begin{theorem}
        The expected number of vertex visits performed by the improved algorithm using the sequence $a_i=i^4$ on an AUSO of dimension~$n$ is at most $e^{\sqrt{2n}+O(n^{1/4}\log n)}$.
    \end{theorem}

This matches the bound of Theorem~\ref{T-main}, up to lower order terms, without the need to adapt the sequence~$A$ to the dimension of the input AUSO.

\section{Concluding remarks and open problems}\label{S-conclude}

We obtained an improved randomized algorithm for finding sinks of AUSOs with an expected running time of $e^{\sqrt{(2+o(1))n}}$. As this is also the upper bound we have on the `fictional' algorithm of Section~\ref{S-fiction}, it seems that this is the best bound that can be obtained using the current techniques. Obtaining an improved algorithm with a run time of $e^{o(\sqrt{n})}$ is of course a major open problem. As mentioned, obtaining an algorithm with a polynomial running time, if one exists, would have far-reaching consequences.

Matou{\v{s}}ek \cite{Matousek94} obtained an $\Omega(e^{\sqrt{2n}}/n^{1/4})$ lower bound on the complexity of the algorithm of Matou{\v{s}}ek, Sharir and Welzl \cite{MaShWe96}, and of G{\"a}rtner \cite{Gartner02}, on some AUSOs. It would be interesting to see if the techniques of Matou{\v{s}}ek~\cite{Matousek94} could be used to give a similar bound, which would be essentially tight, for the improved algorithm presented here. 

Schurr and Szab{\'o} \cite{ScSz04} obtained an $\Omega(n^2/\log n)$ lower bound on the number of queries made in the worst case by any deterministic algorithm for finding sinks of AUSOs. Improving this bound, or extending it to randomized algorithms are also very interesting open problems. 

Another interesting randomized sink-finding algorithm for AUSOs is the \emph{Random-Edge} algorithm. Random-Edge simply performs a random (directed) walk on the AUSO. From each vertex reached it chooses a random outgoing edge. As an AUSO is acyclic, the algorithm must eventually reach the sink. Upper bounding the complexity of Random-Edge is another interesting open problem. The currently best upper bound on its complexity is exponential. (See G{\"a}rtner and Kaibel \cite{GaKa07} and Hansen, Paterson and Zwick \cite{HansenPZ14}.) It is known, however, that Random-Edge is actually slower than Random-Facet, at least on some AUSOs. (See Matou{\v{s}}ek and Szab{\'o} \cite{MaSz06} and Hansen and Zwick \cite{HansenZ16}.)

\paragraph{Note added in proof} OpenAI \cite{OAI:Turn-Based-Stochastic-Mean-Payoff-Games-in-Deterministic-Quasipolynomial-Time-October-5-2026} recently announced a quasipolynomial-time algorithm for solving 2-TBSG. As far as we are aware, this result does not yield an improved algorithm for finding sinks of AUSOs.

\paragraph{Declaration} Interactions with ChatGPT 5.5 Pro led to the formulation and proof of Lemma~\ref{L-G-bound}. The authors assume responsibility for all content.

\newpage
\clearpage
\appendix
\addtolength{\headsep}{0.75cm}   
\addtolength{\textheight}{-1.5cm} 
\addtolength{\topmargin}{-0.5cm} 
\pagestyle{appendix}

\section{Basic Random-Facet related recurrence relations}\label{A-basic-recurrence}

The following two lemmas appear in Hansen and Zwick \cite{HaZw15} as Lemmas~4.2 and~4.3 and Corollary~4.4. For completeness, we repeat their proofs.

\begin{lemma}\label{L-f-exact}
    If $f(n)=f_\alpha(n)$, for $\alpha>0$, is defined as:
    \begin{align*}
    f(n) &\EQ f(n-1) + \frac{\alpha}{n}\sum_{i=1}^n f(n-i) \quad,\quad n>0 \;,\\[-5pt]
    f(0) &\EQ 1 \;.\\[-25pt]
    \end{align*}
    then
    \[ f(n) \EQ \sum_{k=0}^n \frac{\alpha^k}{k!}\binom{n}{k} \quad,\quad n\ge 0\;.\]
\end{lemma}

\begin{proof}
    We let $\bar{f}(n)=\sum_{k=0}^n \frac{\alpha^k}{k!}\binom{n}{k}$ and show that $\bar{f}(n)$ satisfies the recurrence relation defining~$f(n)$. For $n=0$ we have $\bar{f}(0)=f(0)=1$. For $n>0$, using the binomial identity $\binom{n}{k}=\binom{n-1}{k}+\binom{n-1}{k-1}$ and the identities $\binom{n-1}{k}=\frac{k+1}{n}\binom{n}{k+1}$ and $\binom{n}{k+1}=\sum_{j=k}^{n-1}\binom{j}{k}$ \footnote{Note that $\binom{j}{k}$ is the number of choices of~$k+1$ items from $[n]$ in which the largest item is $j+1$.}, we have
        \begin{align*}
        \!\!\!\!\bar{f}(n) - \bar f(n{-}1) 
        &\EQ \sum_{k=0}^n \frac{\alpha^k}{k!}\binom{n}{k} - \sum_{k=0}^{n-1} \frac{\alpha^k}{k!}\binom{n-1}{k} 
         \EQ  \sum_{k=1}^n \frac{\alpha^k}{k!}\left(\binom{n}{k}-\binom{n-1}{k}\right) \\
        &\EQ \sum_{k=1}^n \frac{\alpha^k}{k!}\binom{n-1}{k-1} 
         \EQ \sum_{k=0}^{n-1}\frac{\alpha^{k+1}}{(k+1)!}\binom{n-1}{k}
         \EQ \sum_{k=0}^{n-1}\frac{\alpha^{k+1}}{(k+1)!}\frac{k+1}{n}\binom{n}{k+1} \\
        &\EQ\! \frac{\alpha}{n} \sum_{k=0}^{n-1}\frac{\alpha^k}{k!}\binom{n}{k{+}1}
         \!\EQ\! \frac{\alpha}{n} \sum_{k=0}^{n-1}\frac{\alpha^k}{k!} \sum_{j=k}^{n-1} \binom{j}{k}
         \!\EQ\! \frac{\alpha}{n} \sum_{j=0}^{n-1} \sum_{k=0}^{j} \frac{\alpha^k}{k!} \binom{j}{k}
         \!\EQ\! \frac{\alpha}{n} \sum_{j=0}^{n-1} \bar{f}(j) \;. \qedhere
    \end{align*}
\end{proof}

\begin{lemma}\label{L-f}
    If $f(n)=f_\alpha(n)$, for $\alpha>0$, is defined in Lemma~\ref{L-f-exact}, then
    \[ f(n) \LE e^{2\sqrt{\alpha n}}\;.\]
\end{lemma}

\begin{proof}
    Using the inequality $\binom{n}{k}\le \frac{n^k}{k!}$, the identity $\sum_{k=0}^\infty \frac{x^k}{k!}=e^x$, and the inequality $\sum_{i=0}^n a_i^2 \le (\sum_{i=0}^n a_i)^2$, when $a_i\ge 0$, we have 
    \begin{align*}
        f(n) &\EQ \sum_{k=0}^n \frac{\alpha^k}{k!}\binom{n}{k}
        \LE \sum_{k=0}^n \frac{\alpha^k}{k!}\frac{n^k}{k!}
        \EQ \sum_{k=0}^n \biggl(\frac{(\alpha n)^{k/2}}{k!}\biggr)^2 \\
        &\LE \biggl(\sum_{k=0}^n \frac{(\alpha n)^{k/2}}{k!}\biggr)^2
        \LE \biggl(\sum_{k=0}^\infty \frac{(\alpha n)^{k/2}}{k!}\biggr)^2
        \LE e^{2\sqrt{\alpha n}} \;. \qedhere
        \end{align*}
\end{proof}

We actually have $f_\alpha(n)=L_n(-\alpha)$, where $L_n(x)$ is the $n$-th \emph{Laguerre polynomial} \cite[Chapter 18]{NISTDLMF}. The Laguerre polynomials are defined as follows:
\[ L_n(x) \EQ \sum_{r=0}^n \binom{n}{r}\frac{(-x)^r}{r!} \;.\]
They also satisfy the finite recurrence:
\[ \begin{array}{c} L_{n+1}(x) \EQ \displaystyle\frac{(2n+1-x)L_n(x)-nL_{n-1}(x)}{n+1} \;, \\[10pt]
 L_1(x) \EQ 1-x \quad,\quad L_0(x) \EQ 1 \;. \end{array}\]
It is known that for $\alpha>0$,
\[
f_\alpha(n) \EQ L_n(-\alpha)
\;\sim\;
\frac{e^{-\alpha/2}}{2\sqrt{\pi}\,\alpha^{1/4}n^{1/4}}
\,e^{2\sqrt{\alpha n}}\quad,
\quad  n\to\infty\;.
\]
The simple analysis given in Lemma~\ref{L-f} is thus essentially tight.

We next analyze a shifted version $f_\alpha(n)$ used in the analysis of the improved Random Facet algorithm. 
\begin{lemma}\label{L-g}
    If $g(n)=g_\alpha(n)$, for $\alpha>0$, is defined as:
    \[ g(n) \EQ g(n-1) + \frac{\alpha}{n-1}\sum_{i=1}^{n-1} g(n-i) \quad,\quad n>1\;, \]
    \[ g(1) \EQ 1 \;,\]
    then $g_\alpha(n)=f_\alpha(n-1)$, for $n\ge 1$, where $f_\alpha(n)$ is defined in Lemma~\ref{L-f}. Thus
    \[ g(n) \LE e^{2\sqrt{\alpha (n-1)}} \LT e^{2\sqrt{\alpha n}}\;.\]
\end{lemma}

\begin{proof}
    The relation $g_\alpha(n)=f_\alpha(n-1)$, for $n\ge 1$, is easily proved by induction. As the basis of the induction we have $g(1)=f(0)=1$. The induction step is:
    \[         g(n) \EQ g(n-1) + \frac{\alpha}{n-1}\sum_{i=1}^{n-1}g(n-i)
        \EQ f(n-2) + \frac{\alpha}{n-1}\sum_{i=1}^{n-1}f(n-1-i)
        \EQ f(n-1) \;. \qedhere \]

\end{proof}

\newpage
\restoregeometry          
\pagestyle{plain}         
\printbibliography

\end{document}